\documentclass[11pt]{article}
\usepackage{graphicx} 
\usepackage{tikz}
\usetikzlibrary{patterns,decorations, shapes, calc, positioning, arrows.meta}
\usepackage{pgfplots}
\pgfplotsset{compat=1.9}

\usepackage{subcaption}
\usepackage{amsfonts, amsmath, amssymb, amsthm}
\usepackage{mathtools}
\usepackage{mathabx}
\usepackage[margin=1in]{geometry}
\usepackage{vwcol} 
\usepackage{hyperref}
\hypersetup{
    colorlinks=true,
    linkcolor=purple,
    citecolor=purple,
    urlcolor=blue,
}
\usepackage{wrapfig}
\usepackage[nameinlink]{cleveref}
\usepackage{tgpagella,eulervm}
\usepackage{titling}
\usepackage{microtype}
\usepackage{multicol}
\usepackage{paracol}

\usepackage[noend]{algpseudocode}
\usepackage{adjustbox}
\usepackage{graphicx}
\usepackage{textcomp}
\usepackage{xcolor}
\usepackage{subcaption}
\usepackage[breakable,skins]{tcolorbox}
\usepackage{pifont}
\usepackage{fancybox}
\usepackage{booktabs}

\usepackage{orcidlink}
\usepackage{xparse}

\date{}

\author{Keerthana Gurushankar \orcidlink{0009-0004-4350-6906} \\ 
  \texttt{kgurusha@cs.cmu.edu}
\and 
  Noah G. Singer \orcidlink{0000-0002-0076-521X}\\
  \texttt{ngsinger@cs.cmu.edu}
\and 
Bernardo Subercaseaux \orcidlink{0000-0003-2295-1299}\\
  \texttt{bsuberca@cs.cmu.edu}}

\newtheorem{theorem}{Theorem}
\newtheorem{definition}[theorem]{Definition}
\newtheorem{lemma}[theorem]{Lemma}
\newtheorem{proposition}[theorem]{Proposition}

\crefname{line}{Line}{Lines}
\Crefname{line}{Line}{Lines}
\crefname{lemma}{Lemma}{Lemmas}
\Crefname{lemma}{Lemma}{Lemmas}

\begin{document}

\title{\Huge \sc Latency Guarantees for\\ Caching with Delayed Hits\thanks{Author names alphabetical. \textsc{k.g.} was supported by NSF grant 1942124. \textsc{n.g.s.} was supported in part by an NSF Graduate Research Fellowship (Award DGE 2140739). \textsc{b.s.} was supported by the NSF under grant DMS-2434625.}}

\maketitle

\vspace{-4em}
\noindent\rule{\linewidth}{1pt}

\begin{abstract}
  In the classical caching problem, when a requested page is not present in the cache (i.e., a ``miss''), it is assumed to travel from the backing store into the cache \emph{before} the next request arrives. However, in many real-life applications, such as content delivery networks, this assumption is unrealistic.
  The \emph{delayed-hits} model for caching, introduced by Atre, Sherry, Wang, and Berger, accounts for the latency between a missed cache request and the corresponding arrival from the backing store. This theoretical model has two parameters: the \emph{delay} $Z$, representing the ratio between the retrieval delay and the inter-request delay in an application, and the \emph{cache size} $k$, as in classical caching. Classical caching corresponds to $Z=1$, whereas larger values of $Z$ model applications where retrieving missed requests is expensive. Despite the practical relevance of the delayed-hits model, its theoretical underpinnings are still poorly understood.
  
  We present the first tight theoretical guarantee for optimizing delayed-hits caching: The ``Least Recently Used'' algorithm, a natural, deterministic, online algorithm widely used in practice, is $O(Zk)$-competitive, meaning it incurs at most $O(Zk)$ times more latency than the (offline) optimal schedule. Our result extends to any so-called ``marking'' algorithm.
\end{abstract}
\noindent\rule{\linewidth}{1pt}

\newcommand{\kcache}{k\textsc{-Caching}}
\newcommand{\Zkcache}{(Z,k)\textsc{-DelayedHitsCaching}}
\newcommand{\twocache}{(2,k)\textsc{-DelayedHitsCaching}}
\newcommand{\kserver}{k\textsc{-Servers}}
\newcommand{\kwtcache}{k\textsc{-WeightedCaching}}

\newcommand{\rseq}{\mathbf{r}}
\newcommand{\esched}{\mathbf{e}}
\newcommand{\Alg}{\mathcal{P}}

\newcommand{\POPT}{\Alg_\texttt{OPT}}
\newcommand{\PLRU}{\Alg_\texttt{LRU}}
\newcommand{\PLFU}{\Alg_\texttt{LFU}}
\newcommand{\PFIFO}{\Alg_\texttt{FIFO}}

\newcommand{\LOPT}{L_{\texttt{OPT}}}
\newcommand{\LLRU}{L_{\texttt{LRU}}}

\newcommand{\CR}{\mathsf{CR}}
\newcommand{\cost}{\operatorname{cost}}


\section{Introduction}

Caches are vital components in many modern computing systems. A \emph{cache} is a small but fast store keeping frequently accessed data to allow for quick and efficient access. Whenever data that is not currently in the cache gets requested, that data must be fetched from a larger but much slower memory store, thus resulting in a higher latency. Naturally, minimizing latency is one of the fundamental goals of caching policies.

As caches play an integral role in many memory-intensive and latency-sensitive computer systems~\cite{BSH17,RCK+16}; the problem of designing caching policies to decide which data to cache, and which data to evict from the cache when it is full, have been studied extensively and in various settings~\cite{Wan99,FS17}.

Despite the mathematical beauty and simplicity of the earlier models of caching as an~\emph{online} problem~\cite{ST85, FKL+91, CKN03}, several gaps between those models and practical applications have been identified, giving rise to more realistic (but harder to analyze) models. Some examples are:
\begin{itemize}
    \item \emph{``What if item requests are diverse in nature?''}: while the original caching problem considered uniform memory pages, researchers have studied models in which memory pages can have different sizes as well as different ``weights'' that model how certain pages are more important than others for a given application. In those models, when an important page is requested but is not in the cache, the model factors in a higher cost~\cite{CKPV91, folwarcznyGeneralCachingHard2015, generalCaching}.
    \item \emph{``What about fairness?''}: In multi-user environments,~\emph{caching with reserves} aims to model multiple agents that share a common cache in which each agent gets a guaranteed portion of the space~\cite{cachingWithReserves, reservesViaMarking}.
    \item \emph{``What if we have good predictions about the memory that will be requested in the future?''}: In the traditional online setting, caching policies are assumed to not have any knowledge of future requests. The line of work kick-started by Lykouris and Vassilvitskii~\cite{LV21} incorporates imperfect knowledge about future requests through ``predictions'' (e.g., obtained from machine learning algorithms)~\cite{WeidghtedWithPredictions, antoniadis2022pagingsuccinctpredictions, guptaAugmentingOnlineAlgorithms2022}.
    \item \emph{``What if requests have deadlines?''} In practice, requests might not need to be served immediately, but will rather have associated deadlines, as in the case of the I/O Linux kernel scheduler~\cite{TimeWindows}.
\end{itemize}

In a similar spirit, the \emph{``caching with delayed hits''} model, introduced by Atre, Sherry, Wang and Berger~\cite{ASWB20}, aims to account for a different shortcoming of traditional caching models. Let us illustrate with a simple case scenario.

\begin{tcolorbox}[colframe=blue!50!black, colback=blue!10!white, title=Example 1, rounded corners, breakable, enhanced]
Consider a scenario in which data is being requested at a consistent rate of $100$ requests per second, meaning the inter-request time (IRT) is roughly $0.01$ seconds.  Whenever a request is present in the cache, it can be served very quickly, e.g., $5$ milliseconds, and thus will usually be served before the next request arrives. However, if a request is not in the cache and needs to be fetched from the backing memory store, it takes $0.2$ seconds to be served. Consider now a user $A$ that requests a piece of data $p$ that is present in the cache at time $t$, a user $B$ that requests a different datum $p'$ at time $t + 0.01$ that is not present in the cache at that time, and a third user $C$ that requests the same datum $p'$ at time $t + 0.15$.
As $p$ was present in the cache, it will be served with almost no latency at time $t+0.005$; this is traditionally called a \emph{``hit''}.
Because $p'$ was not present in the cache at the time of its first request, $t + 0.01$, it will take $0.2$ seconds to serve it at time $t + 0.21$. Thus, $B$ will experience the worst-case latency; a so-called \emph{``miss''}. Interestingly, $C$ will have an intermediate experience, since they requested $p'$ at time $t+0.15$, and at time $t+0.21$ their request can be served by benefiting from the fact that $B$ requested the same datum earlier on. This can be easily achieved by implementing a queue that stores the different requests for each datum so they can all be served once a datum arrives from the backing store.
This intermediate latency phenomenon, that occurs when subsequent requests for a datum arrive before the system has time to retrieve the datum from the backing store,  is known as a \emph{``delayed hit''}.
\end{tcolorbox}

In general, when the latency to access the backing memory store is much larger than the IRT (say, by more than an order of magnitude), either because of very high fetching latency (e.g., fetches over large geographical distance) or very high throughput of requests, successive requests for a missing datum may still encounter higher latency while the initial fetch is not completed, resulting in the \emph{``delayed hits''} phenomenon. Traditional caching models, however, would only consider the latency of the first request as a cost to be minimized.

In many modern systems, delayed hits are not merely a technical edge case, but rather have a significant impact on the latency observed. As throughput demands on many systems have grown steadily~\cite{FPM+18}, memory-fetching latency has long remained near its fundamental physical limits. 
Atre, Sherry, Wang, and Berger~\cite{ASWB20} note applications in network software switches where the fetch latency to IRT ratio ($Z$) is in the order of $100$-$1000$x, and devise delayed-hits aware algorithms which achieve up to 45\% reduction in latency compared to the best classical caching algorithms. Zhang \emph{et al.}~\cite{ZTL+22} address applications in content delivery networks (CDNs), where they obtain latency reductions of $9$-$32\%$. In such settings, it is important to analyze and design caching policies that take into account the delayed hits problem.

In terms of theoretical results, Manohar and Williams~\cite{MW20} proved an~$\Omega(Zk)$ lower bound for the competitive ratio of any deterministic online caching policy in the delayed hits setting (where $Z$ is the fetch latency to IRT ratio, and $k$ is the cache size), which is somewhat analogous to the standard $\Omega(k)$ lower bound for caching without delayed hits~\cite{ST85}. However, proving upper bounds in the delayed hits setting is significantly more challenging and to the best of our knowledge, no upper bound that is linear in $Z$ is known. Zhang, Tan, Li, Han, Jiang, and Li~\cite{ZTL+22} claim an $O(Z^{3/2} k)$ upper bound, using an online delayed-hits algorithm which simulates an online, classical weighted caching algorithm; however, their model differs substantially from ours, allowing heterogeneous page sizes, multiple simultaneous arrivals, and a different notion of cache size (which includes pages on the queue except for those which will not eventually enter the cache).

Our main result is that standard ``marking'' based algorithms such as ``Least Recently Used`` (LRU) achieve an optimal competitive ratio of $\Theta(Zk)$, strengthening the analogy to the $\Theta(k)$-competitiveness of LRU in standard caching.



\subsection{Background}

\newcommand{\SQhit}[3]{
    \node[draw, fill=black, text=white, opacity=1, minimum width=0.25cm, minimum height=0.25cm] (#1#2) at (#1, #2) {#3};
}
\newcommand{\SQdelhit}[3]{
    \node[draw, fill=gray, text=white, opacity=1, minimum width=0.25cm, minimum height=0.25cm] (#1#2) at (#1, #2) {#3};
}
\newcommand{\SQmiss}[3]{
    \node[draw, fill=white, text=black, opacity=1, minimum width=0.25cm, minimum height=0.25cm] (#1#2) at (#1, #2) {#3};
}

\newcommand{\Checkmark}[3] {
  \fill[draw, scale=1, fill=#3](#1 + 0.05, #2 + .30) -- (#1 + .26, #2 + 0 ) -- (#1 + 1, #2 + .7) -- (#1 + .25, #2+ .17) -- cycle;
}

\newcommand{\cmark}{\ding{51}}%
\newcommand{\xmark}{\ding{55}}%
\newcommand{\lambdaDist}{1.5}
\tikzset{cylinder end fill/.style={path picture={
\pgftransformshift{\centerpoint}%
\pgftransformrotate{\rotate}%
\pgfpathmoveto{\beforetop}%
\pgfpatharc{90}{-270}{\xradius and \yradius}%
\pgfpathclose
\pgfsetfillcolor{#1}%
\pgfusepath{fill}}
}}
   \def\cacheShift{2.5}
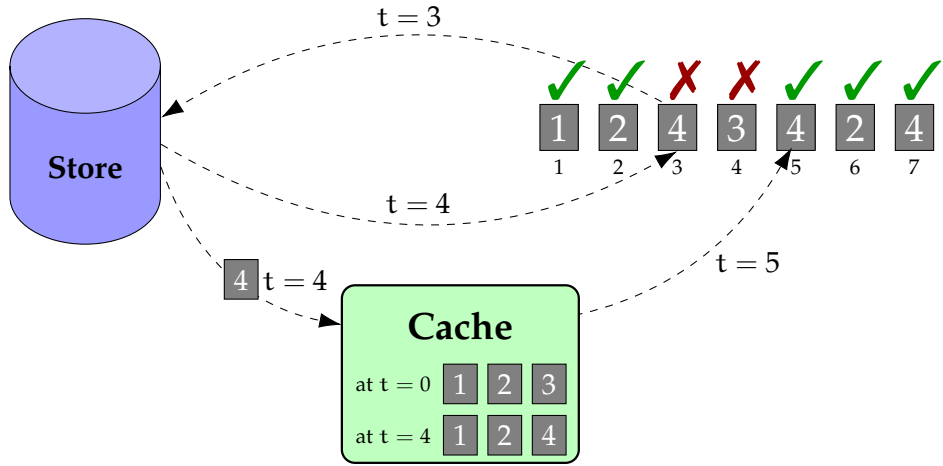
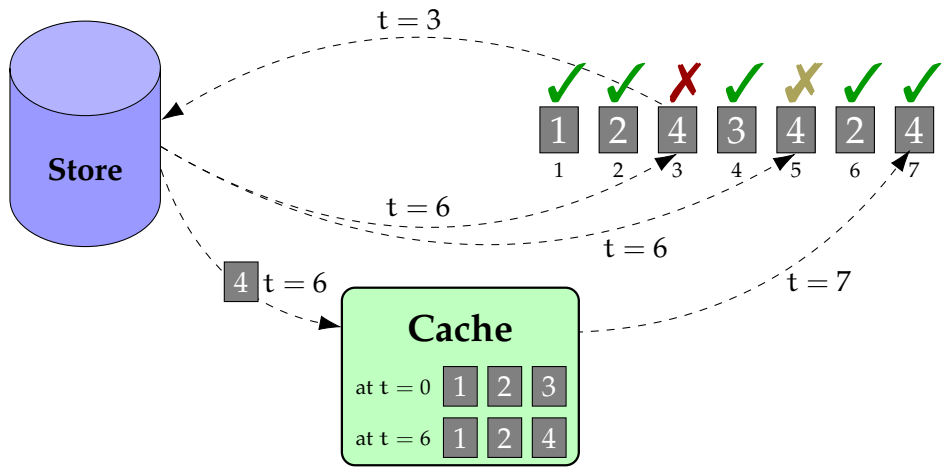
\begin{figure*}
\begin{subfigure}{1.0\textwidth} 
    \centering
    \begin{tikzpicture}[scale=0.525]

    \node (Store) [cylinder, shape border rotate=90, draw,minimum height=3cm,minimum width=2cm, 
    fill=blue!40!white,
    cylinder end fill=blue!30!white] at (0, 0)
{\large \textbf{Store}};

    \def\lambda{1.5}
    \SQdelhit{12}{1.0}{\Large 1};
    \SQdelhit{12+1*\lambda}{1.0}{\Large 2};
    \SQdelhit{12+2*\lambda}{1.0}{\Large 4};
    \SQdelhit{12+3*\lambda}{1.0}{\Large 3};
    \SQdelhit{12+4*\lambda}{1.0}{\Large 4};
    \SQdelhit{12+5*\lambda}{1.0}{\Large 2};
    \SQdelhit{12+6*\lambda}{1.0}{\Large 4};
        \foreach \i in {1,...,7}{
        \node (\i) at (12 + \lambda*\i - \lambda, 0.0) {\scriptsize $\i$};
    }

    \coordinate (RightMiddle) at ($(Store.east)!0.65!(Store.north east)$);
    
    \draw[->, dashed, bend right, >=Latex, >={Latex[width=2mm,length=3mm]}] (12 + 2*\lambda-0.4, 1.65) to node[midway, above] {$t = 3$} (RightMiddle) ;

 \coordinate (RightMiddle-) at ($(Store.east)!0.3!(Store.north east)$);

 \coordinate (RightMiddle--) at ($(Store.east)!0.0!(Store.north east)$);

     \draw[->, dashed, bend right, >=Latex, >={Latex[width=2mm,length=3mm]}, ] (RightMiddle-) to node[midway, above] {$t = 4$} (12 + 2*\lambda, 0.4) ;

      \draw[->, dashed, bend right, >=Latex, >={Latex[width=2mm,length=3mm]}, ] (RightMiddle--) to node[pos=0.55, right, xshift=4pt] {$\, t = 4$} (4+ \cacheShift, -4) ;

      \SQdelhit{1.45+ \cacheShift}{-2.85}{4};

      \draw[->, dashed, bend right, >=Latex, >={Latex[width=2mm,length=3mm]}, ] (10, -4) to node[pos=0.6, right, xshift=5pt] {$t = 5$} (12+4*\lambda-0.1, 0.5) ;

    \draw[rounded corners=5pt, line width=0.3mm, draw=black, fill=green!25!white] (4 + \cacheShift,-3) rectangle (10+ \cacheShift,-7.5);

    \node (cache) at (7+ \cacheShift, -4) {{\Large \textbf{Cache}}};
    \node (cacheText) at (5.3+ \cacheShift, -5.5) {\scriptsize at $t = 0$};
    \node (cacheText2) at (5.3+ \cacheShift, -6.8) {\scriptsize at $t = 4$};
    
    \SQdelhit{7+ \cacheShift}{-5.5}{1};
    \SQdelhit{7+0.75*\lambda+ \cacheShift}{-5.5}{2};
    \SQdelhit{7+1.5*\lambda+ \cacheShift}{-5.5}{3};

    \SQdelhit{7+ \cacheShift}{-6.8}{1};
    \SQdelhit{7+0.75*\lambda+ \cacheShift}{-6.8}{2};
    \SQdelhit{7+1.5*\lambda+ \cacheShift}{-6.8}{4};

    \node[color=green!60!black] (c1) at (12 + \lambda*0 + 0.2 , 1.45 + 0.75) {\huge \cmark};
     \node[color=green!60!black] (c2) at (12 + \lambda*1 + 0.2 , 1.45 + 0.75) {\huge \cmark};
      \node[color=red!60!black] (c3) at (12 + \lambda*2 + 0.2 , 1.45 + 0.75) {\huge \xmark};
      \node[color=red!60!black] (c4) at (12 + \lambda*3 + 0.2 , 1.45 + 0.75) {\huge \xmark};
      \node[color=green!60!black] (c5) at (12 + \lambda*4 + 0.2 , 1.45 + 0.75) {\huge \cmark};
      \node[color=green!60!black] (c6) at (12 + \lambda*5 + 0.2 , 1.45 + 0.75) {\huge \cmark};
      \node[color=green!60!black] (c6) at (12 + \lambda*6 + 0.2 , 1.45 + 0.75) {\huge \cmark};
         \end{tikzpicture}
         \caption{Classical caching. ($Z=1$).}
\end{subfigure}

\vspace{10pt}

\begin{subfigure}{1.0\textwidth} 
    \centering
    \begin{tikzpicture}[scale=0.525]

    \node (Store) [cylinder, shape border rotate=90, draw,minimum height=3cm,minimum width=2cm, 
    fill=blue!40!white,
    cylinder end fill=blue!30!white] at (0, 0)
{\large \textbf{Store}};

    \def\lambda{1.5}
    \SQdelhit{12}{1.0}{\Large 1};
    \SQdelhit{12+1*\lambda}{1.0}{\Large 2};
    \SQdelhit{12+2*\lambda}{1.0}{\Large 4};
    \SQdelhit{12+3*\lambda}{1.0}{\Large 3};
    \SQdelhit{12+4*\lambda}{1.0}{\Large 4};
    \SQdelhit{12+5*\lambda}{1.0}{\Large 2};
    \SQdelhit{12+6*\lambda}{1.0}{\Large 4};
        \foreach \i in {1,...,7}{
        \node (\i) at (12 + \lambda*\i - \lambda, 0.0) {\scriptsize $\i$};
    }

    \coordinate (RightMiddle) at ($(Store.east)!0.65!(Store.north east)$);
    
    \draw[->, dashed, bend right, >=Latex, >={Latex[width=2mm,length=3mm]}] (12 + 2*\lambda-0.4, 1.65) to node[midway, above] {$t = 3$} (RightMiddle) ;

 \coordinate (RightMiddle-) at ($(Store.east)!0.3!(Store.north east)$);

  \coordinate (RightMiddle--) at ($(Store.east)!0.0!(Store.north east)$);

     \draw[->, dashed, bend right, >=Latex, >={Latex[width=2mm,length=3mm]}, ] (RightMiddle-) to node[midway, above] {$t = 6$} (12 + 2*\lambda, 0.4) ;


      \draw[->, dashed, bend right, >=Latex, >={Latex[width=2mm,length=3mm]}, ] (RightMiddle--) to node[pos=0.55, right, xshift=4pt] {$\, t = 6$} (4+ \cacheShift, -4) ;

      \SQdelhit{1.45+ \cacheShift}{-2.85}{4};

       \draw[->, dashed, bend right, >=Latex, >={Latex[width=2mm,length=3mm]}, ] (RightMiddle-) to node[pos=0.75, below, yshift=-2pt] {$t = 6$} (12 + 4*\lambda, 0.4) ;


      \draw[->, dashed, bend right, >=Latex, >={Latex[width=2mm,length=3mm]}, ] (10, -4) to node[pos=0.6, right, xshift=5pt] {$t = 7$} (12+6*\lambda-0.1, 0.5) ;

  \draw[rounded corners=5pt, line width=0.3mm, draw=black, fill=green!25!white] (4 + \cacheShift,-3) rectangle (10 + \cacheShift,-7.5);

    \node (cache) at (7 + \cacheShift, -4) {{\Large \textbf{Cache}}};
    \node (cacheText) at (5.3 + \cacheShift, -5.5) {\scriptsize at $t = 0$};
    \node (cacheText2) at (5.3 + \cacheShift, -6.8) {\scriptsize at $t = 6$};
    
    \SQdelhit{7+ \cacheShift}{-5.5}{1};
    \SQdelhit{7+0.75*\lambda+ \cacheShift}{-5.5}{2};
    \SQdelhit{7+1.5*\lambda+ \cacheShift}{-5.5}{3};

    \SQdelhit{7+ \cacheShift}{-6.8}{1};
    \SQdelhit{7+0.75*\lambda+ \cacheShift}{-6.8}{2};
    \SQdelhit{7+1.5*\lambda+ \cacheShift}{-6.8}{4};

    \node[color=green!60!black] (c1) at (12 + \lambda*0 + 0.2 , 1.45 + 0.75) {\huge \cmark};
     \node[color=green!60!black] (c2) at (12 + \lambda*1 + 0.2 , 1.45 + 0.75) {\huge \cmark};
      \node[color=red!60!black] (c3) at (12 + \lambda*2 + 0.2 , 1.45 + 0.75) {\huge \xmark};
      \node[color=green!60!black] (c4) at (12 + \lambda*3 + 0.2 , 1.45 + 0.75) {\huge \cmark};
      \node[color=yellow!60!black] (c5) at (12 + \lambda*4 + 0.2 , 1.45 + 0.75) {\huge \xmark};
      \node[color=yellow!60!black] (c5) at (12 + \lambda*4 + 0.2 , 1.45 + 0.75) {\huge \cmark};
      \node[color=green!60!black] (c6) at (12 + \lambda*5 + 0.2 , 1.45 + 0.75) {\huge \cmark};
      \node[color=green!60!black] (c6) at (12 + \lambda*6 + 0.2 , 1.45 + 0.75) {\huge \cmark};
         \end{tikzpicture}
         \caption{Caching with delayed hits, $Z = 3$.}
\end{subfigure}









    \caption{Comparison of the classical caching model and the delayed hits model. Each request has its arrival time underneath and a symbol above it representing whether it was a hit (\ding{51}), a miss (\ding{55}), or a delayed hit (\ding{51} + \ding{55}).
    The cache starts out containing $1,2,3$, and in both cases the first request for $4$ is a miss. In both cases, the caching policy decides to evict page $3$ in order to cache page $4$, but the resulting latencies are different as described next. In standard caching, that decision results in a miss for page $3$ at time $4$, and a total latency of $2$. In the delayed hits model, however, page $4$ arrives from the store at time $6$, resulting in a latency of $3$ for the request at $t=3$, and a delayed hit for the request at time $5$ with a latency of $1$, accruing a total of latency $4$. }
    \label{fig:dhits-ex}
\end{figure*}

We begin by providing more precise problem statements for the caching problems described above.

\textbf{The $\kcache$ problem.} Given a universe of $n$ possible memory pages, and a $k$-sized cache, i.e., a subset $S \subseteq [n]$ of size $k$, at every time-step $t$, we get a request $\rseq(t)$, for a page $i\in[n]$. If $i\in S$, we say that the cache has a \emph{``hit''}; otherwise, the cache has a \emph{``miss''} and the item is retrieved from a backing store. The caching algorithm then decides whether or not to cache the newly retrieved item. As this model assumes the cache is always at full capacity (i.e., containing $k$ pages), if the caching algorithm decides to cache the newly retrieved item $\rseq(t)$,  then it must choose an element of $S$ to evict in order to make room for $\rseq(t)$. The objective of caching policies in this model is to minimize the total number of misses.

\begin{tcolorbox}[colframe=blue!50!black, colback=blue!10!white, title=Example 2, rounded corners, left=1pt, right=1pt]
Consider a cache with size $k=3$ and $n=5$ pages. Assume the cache starts by containing pages $1,2,4$, and we receive a request sequence $\rseq = (2, 5, 1, 4, 2, 2)$. Namely, in the first time-step, we receive a request for $2$, and as it is a \emph{hit} we can directly serve it. Next, we receive a request for $5$. As $5$ is not in the cache it is a \emph{miss}. We put $5$ in cache, and must evict one of $1, 2, 4$ in its place. Say, we evict $4$.  Next, at $t=3$, we receive a request for $1$, which is a hit. Then a request for $4$, which would be a miss as we evicted $4$. Say we evict $5$ here. The next two requests, both for $2$, are hits. Thus, our policy would incur 2 misses over this sequence.
\end{tcolorbox}

In 1966, B\'el\'ady~\cite{Bel66} showed that the optimal policy for the $\kcache$ problem is to evict the page which is to be requested \emph{farthest in the future (FIF)}.


However, B\'el\'ady's rule only applies when the complete request sequence $\rseq$ is known in advance. We generally do not have access to this information from the future, and must make caching decisions \emph{online}, with the information available to us at the time of each request. In this online setting, many heuristic caching policies are employed, such as 
\begin{itemize}
    \item $\PLRU$ (evict the Least Recently Used page first),
    \item $\PFIFO$ (First In First Out), and
    \item $\PLFU$ (evict the Least Frequently Used page first).
\end{itemize}
These policies effectively strive to be heuristic estimators of the \emph{time-to-next-access} ranking of pages in cache, in order to mirror the behavior of FIF in obtaining the least cost (number of cache misses). 

While we can deduce performance comparisons of these policies by measuring their cache hit ratios on benchmark workloads (in fact, we present such benchmarks in~\Cref{sec:evaluation}), the conclusions provided by this approach ought to be taken cautiously, since they can strongly vary based on particular patterns in the benchmarking data (e.g., what if the benchmarks suggest policy $X$ is great only because they do not include request sequences containing some rare specific conditions that makes $X$ perform very poorly?). Therefore, our main focus in this article is to gain a theoretical understanding of the worst case of caching policies. Following the standard in the literature, we do this by focusing on the \emph{``competitive ratio''} of caching policies, which corresponds not to the worst case input in terms of the policy's cost but rather to the worst case input in terms of the ratio of how the policy performed and an optimal policy would have performed. 

\begin{definition}[Competitive ratio] The competitive ratio of a caching policy $\Alg$ over a request sequence $\rseq$ is defined as
\[
    \CR(\rseq, \Alg) = \frac{\cost( \rseq, \Alg)}{\cost(\rseq)}, 
\]
where $\cost(\rseq) := \min_\Alg \cost(\rseq, \Alg))$ denotes the cost inherent to the request sequence. The general competitive ratio of the policy is defined as
\[
    \CR(\Alg) = \lim_{m \to \infty}\sup_{|\rseq| = m} \CR(\rseq, \Alg).
\]
\end{definition}

In classical caching, the cost function corresponds simply to the number of misses incurred.
In this model, Sleator and Tarjan~\cite{ST85} showed that $\PLRU$ has a competitive ratio of $O(k)$ in the $\kcache$ problem\footnote{This result extends to a broad class of online policies, known as \emph{marking algorithms}~\cite{OnlineAlgorithms2017a}.}, and that no deterministic online policy could do better. 
Later on, Fiat~\emph{et al.}~\cite{FKL+91} showed that by using randomness one could obtain $O(\log k)$-competitiveness, and that this result was tight.

The delayed hits setting introduced by Atre \emph{et al.}~\cite{ASWB20} examines \emph{latency} (which will be described precisely next) as the cost function. In addition to this, there are a few more significant changes which we outline below.

\textbf{The $\Zkcache$ problem.} 
Let us begin with a rough explanation, that is then formalized as a \emph{finite state machine} (FSM) in \Cref{sec:model} to avoid any ambiguities.
Given $n$ pages and $k$-sized cache $S\subseteq[n]$, at each time-step $t$, we get a request $\rseq(t)=i\in[n]$. If $\rseq(t) \in S$, the cache has a hit, and incurs $0$ latency. On the other hand, if $\rseq(t) \not\in S$, there are two cases: (i) either page $\rseq(t)$ has not been requested in the last $Z$ time steps, in which case it must be requested from the backing store, whereupon it will take $Z$ time steps to reach the cache, resulting in a latency of $Z$, or (ii) that same page $\rseq(t)$ has already been requested some $z < Z$  time steps ago; in that case, the previous request already ordered the page from the backing store, and thus the current request incurs a \emph{delayed hit} with latency $Z-z$. Whether pages have been requested to the backing store or not will be kept in a queue. Once a fetched page arrives from the backing store, a caching policy needs to decide whether to cache it or not, and if it caches, then it must choose some other page to evict from the cache. A small example is outlined in \Cref{fig:dhits-ex}.

There are a few important ways to modify the problem above, such as the \emph{anti-monotonocity} condition noted by Manohar and Williams~\cite{MW20}, as well as the optional \emph{bypassing} feature discussed by Zhang \emph{et al.}~\cite{ZTL+22}. In the interest of clarity, we precisely define the $\Zkcache$ problem as a \emph{finite state machine} (FSM) in \Cref{sec:model}.






\section{Finite-state machine model}\label{sec:model}

\newcommand{\vCache}[1]{S^{(#1)}}
\newcommand{\vStatus}[1]{\sigma^{(#1)}}
\newcommand{\vLatencyChange}[1]{\Delta L^{(#1)}}
\newcommand{\vReq}[1]{r_{#1}}
\newcommand{\vEvict}[1]{e_{#1}}

\newcommand{\sHit}{{\texttt{Hit}}}
\newcommand{\sDHit}{{\texttt{DelayedHit}}}
\newcommand{\sMiss}{{\texttt{Miss}}}
\newcommand{\sEvict}{{\texttt{Eviction}}}

\newcommand{\fSwap}[2]{\texttt{Swap}(#1,#2)}

Here, we present a model of delayed hits caching as a formal \emph{finite state machine}, a comparatively reader-friendly alternative to the  original formulation in terms of an integer linear program (ILP) and minimum-cost multi-commodity flow (MCMCF) problem in \cite[{\S}A.1-2]{ASWB20}.

\columnratio{0.45}
\begin{paracol}{2}
\subsection{Model description}

For $k \in \mathbb{N}$, let $[k]$ denote the natural numbers $\{1,\ldots,k\}$.

Our model takes as input an ``instance'' of the caching problem (a sequence of $T$ requests for pages in $[n]$) and a ``schedule'' of when evictions happen (a sequence of $T$ pages in $[n]$ or $\bot$ symbols representing ``no eviction'') and outputs the total latency resulting from the schedule, or $\infty$ if this schedule is invalid.

At each point in time $t \in [T]$, our model defines two auxiliary objects. These are:
\begin{itemize}
    \item $\vCache{t} \subseteq [n]$ is the \emph{cache} at time $t$, a set of pages. Initially, $\vCache{0} = [k]$.
    \item $\vStatus{t} \in \{\emptyset,\sMiss,\sDHit,\sHit\}$ indicates the status of the request at time $t$. $\emptyset$ is a special initialization symbol and $\vStatus{-(Z-1)},\ldots,\vStatus{0} = \emptyset$ while $\vStatus{t} \in \{\sMiss,\sDHit,\sHit\}$ for all $t \geq 1$. 
\end{itemize}
At each time step, our model produces a quantity $\vLatencyChange{t} \in \mathbb{N}$, the additional latency incurred by the request at time $t$, and the final output of the model is either the total latency $\sum_{t=1}^T \vLatencyChange{t} < \infty$ (a ``valid'' schedule) or $\infty$ (an ``invalid'' schedule).
One syntactic difference between our model and the original model of \cite{ASWB20} is that in our model, the latency of a miss or delayed hit is counted at the time of its request, not at the time of its arrival from the queue into the cache.

\switchcolumn

\begin{tcolorbox}[colframe=blue!50!black, colback=blue!10!white, left=0pt,right=0pt, title=FSM for~$\Zkcache$ ]
\begin{algorithmic}[1]
    \Statex \textbf{Parameters:} $Z,k \in \mathbb{N}$.
    \Statex \textbf{Instance size:} $n,T\in\mathbb{N}$.
    \Statex \textbf{Instance:} $\rseq = (\vReq{1},\ldots,\vReq{T}) \in [n]^T$. 
    \Statex \textbf{Eviction schedule:} $\esched = (\vEvict{1},\ldots,\vEvict{T}) \in ([n] \cup \{\bot\})^T$.
    
    \Statex
    \State $\vCache{0} \gets [k]$\label{line:fsm:init-cache}
    \State $\vStatus{-(Z-1)},\ldots,\vStatus{0} \gets \emptyset$
    
    \Statex

    \For{$t=1,\ldots,T$}
        \State{// Update the cache if evicting}

        \If{$\vStatus{t-Z} = \sMiss$}
            \If{$\vEvict{t} \not\in \vCache{t-1} \cup \{\vReq{t-Z}\}$}
            \State \Return $\infty$
            \EndIf
            
            \State $\vCache{t} \gets (\vCache{t-1} \cup \{\vReq{t-Z}\}) \setminus \{\vEvict{t}\}$\label{line:fsm:update-cache}
        \Else

            \If{$\vEvict{t} \neq \bot$}
            \State \Return $\infty$
            \EndIf
        
            \State $\vCache{t} \gets \vCache{t-1}$
            
        \EndIf

        \Statex

        \State{// Process the new request}
        \If{$\vReq{t} \in \vCache{t}$}\label{line:fsm:hit:check}
        
        \State{$\vStatus{t} \gets \sHit$}
        \State{$\vLatencyChange{t} \gets 0$}\label{line:fsm:hit:penalty}
        
        \ElsIf{$\exists i \in [Z-1]$ s.t. $\vStatus{t-i} = \sMiss$ and $\vReq{t-i} = \vReq{t}$}\label{line:fsm:delayed-hit:check}

        \State{$\vStatus{t} \gets \sDHit$}
        \State{$\vLatencyChange{t} \gets Z-i$}\label{line:fsm:delayed-hit:penalty}

        \Else
        
        \State{$\vStatus{t} \gets \sMiss$}
        \State{$\vLatencyChange{t} \gets Z$}\label{line:fsm:miss:penalty}

        \EndIf

        \Statex
    \EndFor

    \Statex
    \State \Return $\sum_{t=1}^T \vLatencyChange{t}$
\end{algorithmic}
\end{tcolorbox}

An example is provided in \Cref{tab:example}.
\end{paracol}
\newcommand{\sMissSh}{\texttt{M}}
\newcommand{\sHitSh}{\texttt{H}}

\begin{table}
    \centering
    \begin{tabular}{c|c|c|c|c|c|c|c}
       Timestep $t$         & $0$       & $1$       & $2$       & $3$       & $4$       & $5$       & $6$       \\ \hline
       Request $r_t$        & ---       & $2$       & $1$       & $2$       & $1$       & $2$       & $1$       \\
       Eviction $\vEvict{t}$& ---       & $\bot$    & $\bot$    & $1$       & $\bot$    & $\bot$    & $2$       \\ \hline
       Cache $\vCache{t}$   & $\{1\}$   & $\{1\}$   & $\{1\}$   & $\{2\}$   & $\{2\}$   & $\{2\}$   & $\{1\}$   \\
       Status $\vStatus{t}$ & $\emptyset$& $\sMissSh$& $\sHitSh$ & $\sHitSh$ & $\sMissSh$& $\sHitSh$ & $\sHitSh$\\
       Latency $\vLatencyChange{t}$ &---& $2$       & $0$       & $0$       & $2$       & $0$       & $0$
    \end{tabular}
    \caption{An example of running the delayed-hits finite state machine on a request sequence of length $T=6$, with delay $Z=2$, cache size $k=1$, on a universe of size $n=2$. In this particular example, every request is either a $\sMiss$ ($\sMissSh$ for short) or a $\sHit$ ($\sHitSh$ for short). Interestingly, in the traditional \emph{no-delay} caching model (i.e., $Z=1$), this request sequence is ``harder'': Every schedule results in at least $3$ misses. (One schedule with exactly $3$ misses is simply to keep $1$ in the cache the entire time.) Finally, we remark that after the $T=6$ timesteps, the cache has reset to its original state and there are no $\sMiss$es in the last $2$ timesteps, so this request sequence and eviction schedule may be repeated and the state is periodic.}
    \label{tab:example}
\end{table}

\subsection{Model properties}

We highlight several important properties of the model. First, we observe that a schedule is valid if and only if the following two conditions hold:

\begin{itemize}
    \item For every $t \in [T]$, if $\vEvict{t} \neq \bot$ then $\vStatus{t-Z} = \sMiss$.
    \item If $\vStatus{t-Z} = \sMiss$, then $\vEvict{t} \in \vCache{t-1} \cup \{\vReq{t-Z}\}$.
\end{itemize}


\begin{proposition}\label{prop:dhit-to-miss}
    For every $t \in [T]$, there can be at most one index $i \in [Z-1]$ such that $\vStatus{t-i} = \sMiss$ and $\vReq{t-i} = \vReq{t}$.
\end{proposition}

\begin{proof}
    Indeed, suppose there were two such distinct indices $i_1 < i_2 \in [Z-1]$. Then $t-i_1$ must be a $\sDHit$ (i.e., $\vStatus{t-i_1} = \sDHit$), since $\vStatus{(t-i_1)-(i_2-i_i)} = \sMiss$ and $\vReq{(t-i_1)-(i_2-i_1)} = \vReq{t-i_1} = \vReq{t} = \vReq{t-i_2}$.
\end{proof}

Thus, there can be at most one index $i \in [Z-1]$ satisfying the condition for the check for a $\sDHit$ on \Cref{line:fsm:delayed-hit:check}, and there is a function mapping each $\sDHit$ to an associated $\sMiss$ within the last $Z$ time steps.

\begin{proposition}\label{prop:cache-size}
    For all $t \in [T]$, if $\vStatus{t-Z} = \sMiss$, then $\vReq{t-Z}\not\in\vCache{t-1}$.
\end{proposition}

\begin{proof}
    Since $\vStatus{t-Z} = \sMiss$, we know $\vReq{t-Z} \not\in \vCache{t-Z}$. Suppose for contradiction that $\vReq{t-Z} \in \vCache{t-1}$. Then there is some $i \in [Z-1]$ such that $\vReq{t-Z} \not\in \vCache{t-i-1}$ but $\vReq{t-Z} \in \vCache{t-i}$. For this to happen, we need $\vReq{t-i-Z} = \vReq{t-Z}$ and $\vStatus{t-i-Z} = \sMiss$. But this, in turn, guarantees that $\vStatus{t-Z} = \sDHit$.
\end{proof}

As a corollary, note that if $\vStatus{t-Z} = \sMiss$, then $|\vCache{t}| = |(\vCache{t-1} \cup \{\vReq{t-Z}\})\setminus\{\vEvict{t}\}|$ which is $|\vCache{t-1}|$ by the proposition. Otherwise, $|\vCache{t}| = |\vCache{t-1}|$ trivially; thus, $|\vCache{t}| = k$ for all $t$ since initially $\vCache{0} = [k]$.

Other important properties include:

\begin{itemize}
    \item In the case $Z=1$, our model recovers the usual model of caching without delayed hits. Indeed, the \textbf{if} statement on \Cref{line:fsm:delayed-hit:check} is vacuous since $[Z-1]=\emptyset$. Thus, when the request at time $t$ is processed, it is always a $\sHit$ or a $\sMiss$, depending on whether the requested page was in the queue. In the time step $t+1$, before request $t+1$ is processed, we can perform an eviction if request $t$ was a $\sMiss$. In this case, we add $\vReq{t}$ to the cache, evicting some chosen element.
    \item For each $t \in [T]$, the values of $\vCache{t}$ and $\vStatus{t}$ are determined using only current values $\vReq{t}$ and $\vEvict{t}$ as well as values $\vCache{t-1}$, $\vStatus{t-1},\ldots,\vStatus{t-(Z-1)}$, and $\vReq{t-1},\ldots,\vReq{t-(Z-1)}$ \emph{from the past $Z$ timesteps}. Indeed, we could view the model as an automaton: At each step, the model takes an input $(\vReq{t},\vEvict{t})$ and previous state
    \[
    (\vCache{t-1}, \vStatus{t-1},\ldots,\vStatus{t-(Z-1)},\vReq{t-1},\ldots,\vReq{t-(Z-1)})
    \]
    and produces a new state
    \[
    (\vCache{t}, \vStatus{t}, \ldots,\vStatus{t-(Z-2)},\vReq{t},\ldots,\vReq{t-(Z-2)}).
    \]
    For fixed $Z$, $k$, and $n$, there are only finitely many possible states. This justifies our ``finite state machine'' terminology.
\end{itemize}

\subsection{Policies and optimization}

Given our definition of the model, we can now define the problem of optimizing delayed hits caching schedules:

\begin{definition}[$\Zkcache$]\label{def:dhits}
    For $Z,k\in\mathbb{N}$, the \emph{delayed hits caching problem with delay $Z$ and cache size $k$}, denoted $\Zkcache$, is defined as follows. Given a number of \emph{pages} $n$ and a number of \emph{time steps} $T$, an \emph{instance} of $\Zkcache$ is a sequence of \emph{requests} for pages $\rseq = (\vReq{1},\ldots,\vReq{T}) \in [n]^T$. An \emph{eviction schedule} is a corresponding sequence of pages $\esched = (\vEvict{1},\ldots,\vEvict{T}) \in ([n] \cup \{\bot\})^T$. The \emph{latency} of the schedule $\esched$ on the instance $\rseq$, denoted $\cost(\rseq,\esched)$, is defined as the value returned by the finite state machine above. The \emph{optimal latency} of the instance $\rseq$, denoted $\cost(\rseq)$, is the minimum latency of any schedule, $\min_\esched \cost(\rseq,\esched)$.
\end{definition}

For a number of pages $n$ and a number of time steps $T$, a (deterministic) \emph{caching policy} is a function $\Alg : [n]^T \to ([n] \cup \{\bot\})^T$ mapping an instance to an eviction schedule. Thus, a caching policy is a purported solution to $\Zkcache$. We are interested in two specific policies:

\begin{definition}[Optimal policy]
    For $n,T \in \mathbb{N}$, the \emph{optimal policy} $\POPT : [n]^T \to ([n] \cup \{\bot\}^T)$, given any request sequence $\rseq \in [n]^T$, outputs an arbitrary valid schedule $\POPT(\rseq)$ of minimal latency $\cost(\rseq,\POPT(\rseq)) = \cost(\rseq)$.
\end{definition}

\newcommand{\tlast}[3]{t_{\texttt{last}}(#1,#2,#3)}

Given a request sequence $\rseq \in [n]^T$, a page $p \in [n]$, and a time $t \in [T]$, let $\tlast{\rseq}{p}{t} := \max\{t' < t : \vReq{t'} = p\}$ denote the last time before $t$ that $p$ was requested (or $-\infty$ if no such time exists).

\begin{definition}[LRU policy]
    For $n,T \in \mathbb{N}$, the \emph{Least-Recently-Used (LRU) policy} $\PLRU : [n]^T \to ([n] \cup \{\bot\})^T$, given any request sequence $\rseq \in [n]^T$, outputs the eviction sequence $\PLRU(\rseq) = (\vEvict{1},\ldots,\vEvict{T})$ defined inductively as follows. At time $t = 1,\ldots,T$, $\vEvict{t} = \bot$ if $\vStatus{t-Z} \neq \sMiss$. Otherwise, $\vEvict{t}$ is a page $s \in \vCache{t-1} \cup \{\vReq{t-Z}\}$ minimizing $\tlast{\rseq}{s}{t}$ among all such pages.
\end{definition}

The LRU policy is \emph{online} in the sense that the eviction schedule's entry at time $t$ is determined only by the request sequence up to time $t$, and not the entire request sequence.

\section{Marking-based upper bound}\label{sec:lru-bound}

\newcommand{\claimmark}{($*$)}
\newcommand{\eLRU}{\esched_{\texttt{LRU}}}
\newcommand{\eOPT}{\esched_{\texttt{OPT}}}

This section states and proves our main theorem.

\begin{tcolorbox}[colback=green!10!white]
\begin{theorem}\label{thm:lru-bound}
    In $\Zkcache$, $\CR(\PLRU)=O(kZ)$.
\end{theorem}
\end{tcolorbox}
 
 To do this, throughout this section, fix $Z,k,n,T \in \mathbb{N}$ and a request sequence $\rseq \in [n]^T$. This fixes $\eLRU := \PLRU(\rseq)$ and $\eOPT := \POPT(\rseq)$, the eviction schedules produced by the LRU and optimal policies, respectively. In turn, this fixes $\LLRU := \cost(\rseq,\eLRU)$ and $\LOPT := \cost(\rseq,\eOPT)$, the respective costs of these schedules. The competitive ratio of $\PLRU$ on the instance $\rseq$ is \[ \CR(\rseq,\PLRU) = \frac{\LLRU}{\LOPT}. \] So, it suffices to prove an upper bound on $\LLRU$ and a lower bound on $\LOPT$. We now introduce two ideas required for our analysis.

\subsection{Phases}

To prove \Cref{thm:lru-bound}, our analysis begins by partitioning the set of time steps $[T]$ into successive \emph{phases} $P^1,\ldots,P^\ell$. This is a standard step in analyzing caching algorithms (see e.g. \cite{ST85}). Phases, as a partition of $[T]$, can be defined by iteratively assigning each time step $t \in [T]$ to a phase $P^i$ as follows:

\begin{itemize}
    \item If $t = 1$, then assign it to $P^1$. 
    \item For $1 < t \leq T$, if $t-1$ was assigned to phase $P^i$, then check whether the set of requests $\vReq{t'}$ at times $t'$ already assigned to $P^i$, together with the request $\vReq{t}$, has size $k+1$. If so, assign $t$ to phase $P^{i+1}$, if not, to $P^{i}$.  
\end{itemize}

Note that the partition of $[T]$ into phases depends only on the cache size $k$ and the request sequence $\rseq$, and \emph{not} on the delay $Z$ or on the eviction schedule.

\begin{tcolorbox}[colframe=blue!50!black, colback=blue!10!white, title=Example 3, rounded corners, left=1pt, right=1pt]
Consider a cache size of $k=2$ and the request sequence  \[ \rseq = (1,2,1,3,3,2,1). \] Then, the $T=7$ time steps are partitioned into phases $P^1 = \{1,2,3\}$, $P^2=\{4,5,6\}$, and $P^3 = \{7\}$. 
\end{tcolorbox}

We say a page $p \in [n]$ is \emph{requested} in phase $i$ if there exists $t \in P^i$ such that $\vReq{t} = p$. In every phase except the last, exactly $k$ distinct pages are requested; in particular, every phase except the last contains at least $k$ time steps. Page $p$ is \emph{fresh} for phase $i \geq 2$ if it is requested in phase $i$ but not phase $i-1$. The first request in a phase $i \geq 2$ is always for a fresh page.

\subsection{Superphases}

Now, we introduce a novel step in the analysis: A further partition of the phases into \emph{superphases}. This partition into superphases depends only on the number of phases $\ell$, the length $|P^i|$ of each phase, and the delay $Z$. The superphases $Q^1,\ldots,Q^m$ partition $[\ell]$ and are defined iteratively:

\begin{itemize}
    \item $Q^1$ begins with phase $1$.
    \item We accumulate phases in $Q^j$ until the total length of the phases in $Q^j$ is at least $Z$. Then, we add two additional phases. The next phase is then the first phase in $Q^{j+1}$. (We stop immediately when we run out of phases.)
\end{itemize}

\begin{tcolorbox}[colframe=blue!50!black, colback=blue!10!white, title=Example 4, rounded corners, left=1pt, right=1pt]
Consider  $Z=10$ and  $\ell=9$ phases with lengths
\[
|P^1|=6,|P^2|=2,|P^3|=3,|P^4|=15,|P^5|=1, |P^6|=11,|P^7|=3,|P^8|=9,|P^9|=1.
\]
Then, the corresponding superphases are:
\[Q^1=\{1,2,3,4,5\}, \; Q^2=\{6,7,8\}, \; Q^3=\{9\}.\]
\end{tcolorbox}

Note that the partition of phases into superphases depends only on the lengths of the phases and the delay $Z$, and not on the caching policy, the cache size, or the request sequence.

\subsection{Proving \Cref{thm:lru-bound}}

Now, we turn to proving \Cref{thm:lru-bound}. As in the previous section, we partition $[T]$ into $\ell$ phases $P^1,\ldots,P^\ell$, and $[\ell]$ further into $m$ superphases $Q^1,\ldots,Q^m$. For $j \in [m]$, let $\LOPT^j$ and $\LLRU^j$ denote the latencies
\[
\sum_{i \in Q^j} \sum_{t \in P^i} \vLatencyChange{t}
\]
incurred in superphase $j$ when running $\eOPT$ and $\eLRU$, respectively. Thus, \[
\sum_{j=1}^m \LOPT^j = \LOPT
\]
and similarly for $\LLRU$.

The proof of \Cref{thm:lru-bound} then relies on the following two complementary lemmas:

\begin{lemma}[$\eLRU$ superphase upper bound]\label{lem:plru-ub}
    For every $j \in [m]$, \( \LLRU^j \leq 4kZ^2. \)
\end{lemma}

\begin{lemma}[$\eOPT$ superphase lower bound]\label{lem:popt-lb}
    For every $j \in [m-1]$, \( \LOPT^j \geq Z. \)
\end{lemma}

We now prove \Cref{thm:lru-bound}:

\begin{proof}[Proof of \Cref{thm:lru-bound}]
Summing the inequality in \Cref{lem:plru-ub} over all $m$ superphases gives $\LLRU \leq 4mkZ^2$, while summing the inequality in \Cref{lem:popt-lb} over the first $m-1$ superphases gives $\LOPT \geq (m-1) Z$. If $m \geq 2$, then $\frac{m}{(m-1)} \leq 2$ and thus \[ \frac{\LLRU}{\LOPT} \leq \frac{4mkZ^2}{(m-1)Z} \leq 8kZ, \] as desired. Otherwise, $m=1$, in which case $\LLRU \leq 4kZ^2,$ and it is easy to see that $\LOPT \geq Z$, since the first time any page in $[n]\setminus[k]$ is requested is always a $\sMiss$ in every policy. (If only pages in $[k]$ are ever requested, both $\eOPT$ and $\eLRU$ will incur zero latency.)
\end{proof}

To complete the proof of~\Cref{thm:lru-bound}, it remains to prove~\Cref{lem:plru-ub,lem:popt-lb}.

\subsection{LRU policy upper bound}

\Cref{lem:plru-ub}, which bounds the latency incurred in each superphase, follows from the following lemma which bounds the latency incurred in each \emph{phase}:

\begin{lemma}\label{lem:plru-ub:phase}
    For every $i \in [\ell]$, in $\eLRU$, \[
    \sum_{t \in P^i} \vLatencyChange{t} \leq kZ \min\{Z,|P^i|\}.
    \]
\end{lemma}

First, we prove \Cref{lem:plru-ub} using \Cref{lem:plru-ub:phase}:

\begin{proof}[Proof of \Cref{lem:plru-ub}]
To prove \Cref{lem:plru-ub}, let $i_1$ and $i_2$ be first and last phases in superphase $Q^j$. Then
\begin{align*}
    \LLRU^j &= \sum_{i = i_1}^{i_2} \sum_{t\in P^i} \vLatencyChange{t} \tag{def. $\LLRU^j$} \\
    &\leq \sum_{i=i_1}^{i_2} kZ \min\{Z,|P^i|\} \tag{\cref{lem:plru-ub:phase}} \\
    &\leq kZ \left( \sum_{i=i_1}^{i_2-3} \min\{Z,|P^i|\} + 3Z \right).
\end{align*}

Finally, we observe that $\sum_{i=i_1}^{i_2-3} \min\{Z,|P^i|\} \leq Z$, since by the definition of superphase we have $\sum_{i=i_1}^{i_2-2} |P^i| \geq Z$ but $\sum_{i=i_1}^{i_2-3} |P^i| < Z$. Thus, $\LLRU^j \leq 4kZ^2$, as desired.
\end{proof}

After the proof of \Cref{lem:plru-ub} from \Cref{lem:plru-ub:phase}, we turn to the proof of \Cref{lem:plru-ub:phase}. This proof does not use the definition of $\PLRU$ directly: Instead, it uses a property of the schedule $\PLRU$ produces, as follows. An eviction schedule $\esched = (\vEvict{1},\ldots,\vEvict{T})$ is \emph{marking} with respect to the request sequence $\rseq = (\vReq{1},\ldots,\vReq{T})$ if for every $i \in [L]$ and every $t<t' \in P^i$, $\vReq{t} \neq \vEvict{t'}$. That is, a page is never evicted during $P^i$ after it has been requested in $P^i$.

\begin{lemma}
    $\eLRU$ is marking.
\end{lemma}

\begin{proof}
    Let $t < t' \in P^i$ and let $p = \vReq{t}$. Without loss of generality, we assume $t$ is the \emph{first} time $p$ is requested during phase $i$.

    Now suppose $\vEvict{t'} = p$. Thus, by definition of $\PLRU$, $\vStatus{t'-Z} = \sMiss$ and $\vEvict{t'}$ is some page $s$ minimizing $\tlast{\rseq}{s}{t'}$ over $s \in \vCache{t'-1} \cup \{\vReq{t'-Z}\}$. Now by \Cref{prop:cache-size}, $|\vCache{t'-1} \cup \{\vReq{t'-Z}\}| = k+1$. Since at most $k$ distinct pages are requested in $P^i$, there is some $s \in \vCache{t'-1} \cup \{\vReq{t'-Z}\}$ not requested in $P^i$. Thus, $\tlast{\rseq}{s}{t'} < \tlast{\rseq}{p}{t'}$ and so $\PLRU$ will not evict $p$.
\end{proof}

\begin{lemma}
    Let $p$ be any page requested during $P^i$, and let $t_1$ denote the time of its first request during $P^i$. Let $\esched$ be any marking schedule. Let
    \[
    t_2 = \begin{cases} t_1 & \vStatus{t_1} = \sHit \\ t_1 + Z & \vStatus{t_1} = \sMiss \\ t_1 + (Z-i(t_1)) & \vStatus{t_1} = \sDHit \end{cases}
    \]
    where in the case $\vStatus{t_1} = \sDHit$, $i(t_1)$ denotes the value $i \in [Z-1]$ such that $\vStatus{t-i} = \sMiss$ and $\vReq{t_1-i} = p$ (unique by \Cref{prop:dhit-to-miss}). Then for all $t' \geq t_2 \in P^i$, $\vReq{t_1} \in \vCache{t'}$.
\end{lemma}

\begin{proof}
    If $t_2 \not\in P^i$, we have nothing to prove. If $t_2 \in P^i$, we proceed inductively. First, we claim that $\vReq{t_1} \in \vCache{t_2}$:
    \begin{itemize}
        \item If $\vStatus{t_1} = \sHit$, then $t_1 = t_2$ and the claim holds trivially.
        \item If $\vStatus{t_1} = \sMiss$, then $t_2 = t_1+Z$. Since $t_2 - Z = t_1$, $\vStatus{t_2 - Z} = \sMiss$, so there is an eviction at time $t_2$. Moreover, $\vReq{t_2 - Z} = p$, and by the definition of marking, $\vEvict{t_2} \neq \vReq{t_1} = p$. Thus, $p \in \vCache{t_2}$ (i.e., $p$ enters the cache and it is not immediately evicted).
        \item If $\vStatus{t_1} = \sDHit$, then $t_2 = t_1 + (Z-i(t_1))$, and $\vStatus{t_1 - i(t_1)} = \sMiss$ and $\vReq{t_1 - i(t_1)} = p$. Since $t_2 - Z = t_1 - i(t_1)$, $\vStatus{t_2 - Z} = \sMiss$, so there is an eviction at time $t_1 + (Z - i(t_1))$. Moreover, $\vReq{t_2 - Z} = p$, and by the definition of marking, $\vEvict{t_2} \neq \vReq{t_1} = p$.\footnote{Note that $t_1 - i(t_1)$ might not be in $P^i$, so we cannot apply marking to $t_1 - i(t_1)$. But, we already know a request for $p$ in $P^i$ earlier than $t_2$, namely, $t_1$.} Thus, $p \in \vCache{t_2}$ (i.e., $p$ enters the cache and it is not immediately evicted).
    \end{itemize}
    Now the definition of marking further implies that $p$ will never be evicted over the remainder of $P^i$, as desired.
\end{proof}
 
Finally, we prove \Cref{lem:plru-ub:phase} directly:

\begin{proof}[Proof of \Cref{lem:plru-ub:phase}]
    By definition, during $P^i$, there are requests for (at most) $k$ distinct pages. So, it suffices to show that for each $p \in [n]$, the latency incurred by requests for $p$ during phase $P^i$ is at most $Z \min\{Z, |P^i|\}$. We now analyze all such requests.
    
    Let $u$ denote the maximum time step in $P^i$.
    Fix $p$, let $t_1$ denote the first time $p$ is requested in $P^i$, and let $t_2$ be defined as in the previous lemma. Then there is no request for $p$ during $P^i$ before $t_1$. Further, since $\eLRU$ is marking, every request for $p$ at time $t' \geq t_2$ during $P^i$ is a $\sHit$, since $p \in \vCache{t'}$. So, to bound the latency, it suffices to consider requests for $p$ at time $t_1 \leq t < \min\{t_2,u\}$. Such a request incurs latency $Z-(t-t_1) \leq Z$. Finally, we observe that \[ \min\{t_2,u\}-t_1 = \min\{t_2-t_1,u-t_1\} \leq \min\{Z, |P^i|\}, \] as desired.
\end{proof}

\subsection{Optimal policy lower bound}

Now, we prove \Cref{lem:popt-lb}. In this subsection, we always assume we are executing the schedule $\eOPT$. Since a $\sMiss$ incurs a latency penalty of $Z$ (\Cref{line:fsm:miss:penalty}), \Cref{lem:popt-lb} follows immediately from the following lemma:

\begin{lemma}\label{lem:popt-lb:phase}
    For every $j \in [m-1]$, there is a $\sMiss$ in some phase $i \in Q^j$.
\end{lemma}

\begin{proof}[Proof of \Cref{lem:popt-lb:phase}]
Let $i_1$ and $i_2$ denote, respectively, the minimum and maximum elements of $Q^j$. Note that $i_2 \geq i_1+2$ and $\sum_{i=i_1}^{i_2-2} |P^i| \geq Z$. We show that

\begin{quote}
    \claimmark: There was either a $\sMiss$, $\sDHit$, or an $\sEvict$ during phases $i_2$ or $i_2-1$.
\end{quote}

If there is a $\sMiss$ in phase $i_2$ or $i_2-1$, we are immediately done. Otherwise, a $\sDHit$ or $\sEvict$ imply a $\sMiss$ at most $Z$ time steps before a time step in $P^{i_2} \cup P^{i_2-1}$; since $\sum_{i=i_1}^{i_2-2} |P^i| \geq Z$, this $\sMiss$ is guaranteed to fall \emph{within} $Q^j$.

Let $p$ be any page which is fresh in phase $i_2$. Let $t$ be the first time $p$ is requested in phase $i_2$ and let $u$ denote the final time step in phase $i_2-1$. We consider cases on the status of the request at time $t$:

\begin{enumerate}
    \item If there was a $\sMiss$ at time $t$, then \claimmark.
    \item If there was a $\sDHit$ at time $t$, then \claimmark.
    \item If there was a $\sHit$ at time $t$, then $p \in \vCache{t}$. This is ``unexpected'' because $p$ was \emph{not} requested in phase $i_2-1$, yet it was in the cache during phase $i_2$. We split into subcases based on whether $p$ was in the cache by the end of phase $i_2-1$ (i.e., $p \in \vCache{u}$):
    \begin{enumerate}
        \item If $p \not\in \vCache{u}$, then since $p \in \vCache{t}$, there was an $\sEvict$ at some time $u < t' \leq t$, so \claimmark.
        \item If $p \in \vCache{u}$, we recall that there were $k$ pages requested during phase $i_2-1$. Since $p$ is fresh for phase $i_2$, it was \emph{not} requested during phase $i_2-1$. Thus, there exists a page $q$ requested in phase $i_2-1$ which was \emph{not} on the cache at the end of phase $i_2-1$ (i.e., $q \not\in \vCache{u}$). Let $s$ be the first time $q$ was requested during phase $i_2-1$. Consider cases based on the status of the request at time $s$:
        \begin{enumerate}
            \item If there was a $\sMiss$ at time $s$, then \claimmark.
            \item If there was a $\sDHit$ at time $s$, then \claimmark.
            \item If there was a $\sHit$ at time $s$, then $q \in \vCache{s}$. Since $q \not\in \vCache{u}$, there is an $\sEvict$ during phase $i_2-1$, so \claimmark.
        \end{enumerate}
    \end{enumerate}
\end{enumerate}

\end{proof}

\section{Evaluation}
\label{sec:evaluation}

A natural question after our analysis is how the behavior of LRU in practice compares to our guaranteed worst-case competitive ratio. In a nutshell, we draw the following conclusion from our experiments:
\begin{itemize}
    \item The competitive ratio does worsen as the cache size $k$ grows larger, albeit at a much smaller rate.
    \item The competitive ratio does not seem to worsen consistently as $Z$ grows larger.
    \item Overall, the competitive ratio in practice seems to always be significantly smaller than the worse-case bound would predict.
\end{itemize}

Let us now proceed to detail our experiments;~\Cref{sec:discussion} presents a discussion of further research direction that follow from the conclusions outlined above.

\paragraph{Datasets}

Our experiments use a total of 18 datasets, 13 of which are drawn from diverse industry applications, and 5 of which are synthetic, following standard practices of caching benchmarks~\cite{cache2kCache2kJava, LV21, sadek2024algorithmscachingmtsreduced}. A general description of the datasets is presented in~\Cref{tab:datasets}. 
Synthetic datasets were generated by~\emph{Zipfian} distributions with different parameters $\alpha$; using the probability-density function
\(
    p_\alpha(m) = \frac{m^{-\alpha}}{\zeta(m)}
\)
for integer $m \geq 1$, and $\zeta$ being Riemann's zeta function. We used parameters $\alpha \in (1.3, 1.5, 1.7, 1.9, 2.1)$. We \emph{``normalized''} every dataset to consist of a total of $T = 5000$ requests, filtering by the ``pages''\footnote{Due to the diverse nature of the datasets, we interpret all the dataset-dependent type of requests (e.g., URLs, hashes, etc.) as pages.} requested most often.

\begin{table}
    \centering
    \begin{tabular}{r|r|r|r}
        \toprule
         Category & \# traces  & $n$ (distinct pages) & Source \\
         \midrule
          \texttt{WikiMedia} & 3  & [461, 480, 455] & \cite{Wikimedia}\\
          \texttt{Bright Kite} & 1 &  66 & \cite{BrightKite}\\
          \texttt{WikiBench} & 4 & [476, 466, 466, 455] & \cite{Wikibench}\\
          \texttt{WebSearch} & 3 & [500, 497, 494] & \cite{Umass}\\
          \texttt{Financial} & 2 & [128, 380] & \cite{Umass}\\
          \texttt{Zipf} & 5 & [985, 427, 205, 114, 79] & Synthetic\\
         \bottomrule
    \end{tabular}
    \caption{Description of the datasets used.}
    \label{tab:datasets}
\end{table}

\begin{figure}
    \centering
    \input{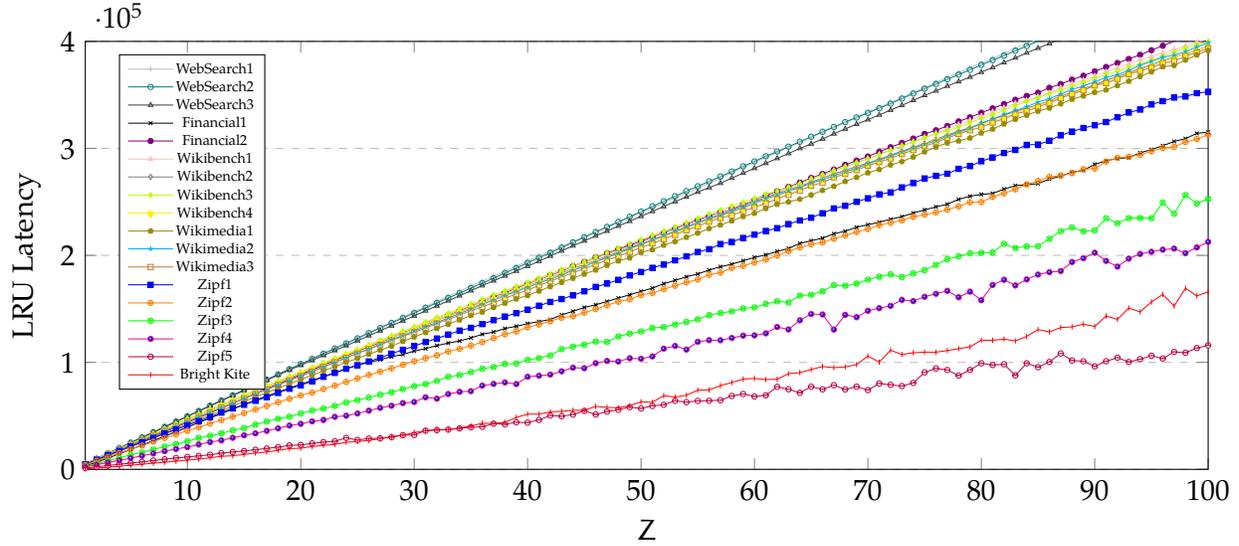}
    \caption{Total latency for LRU as a function of the delay $Z$, with $k = 5$ over all datasets.}
    \label{fig:LRU_v_Z}
\end{figure}
\begin{figure}
    \centering
    \input{figures/cr_plot_tikz}
    \caption{Empirical competitive ratio of LRU (as a function of $Z$) over a WikiBench dataset. Results for other datasets were similar.}
    \label{fig:CR_v_Z}
\end{figure}
\begin{figure}
    \centering
        \begin{tikzpicture}
        \begin{axis}[
            every axis plot/.append style={line width=0.575pt},
            title=,
            width=\linewidth,
            height=\axisdefaultheight,
            xlabel=$k$,
            ylabel=$\frac{\LLRU}{\LOPT}$,
            xmin=1, xmax=30,
            ymin=1, ymax=1.6,
            legend pos=north west,
            ymajorgrids=true,
            grid style=dashed,
            mark size=1pt,
 legend style={nodes={scale=0.5, transform shape}}
        ]
            
        \addplot[
            color=red,
            mark=*,
        ]
        coordinates {
            (1.0, 1.0781284004352558)
(2.0, 1.1341801385681294)
(3.0, 1.1341801385681294)
(4.0, 1.1762285160977972)
(5.0, 1.2118772018117765)
(6.0, 1.2434827945776852)
(7.0, 1.2434827945776852)
(8.0, 1.2735341581495427)
(9.0, 1.297147604541678)
(10.0, 1.3197724039829304)
(11.0, 1.3197724039829304)
(12.0, 1.3412142440163457)
(13.0, 1.3618677042801557)
(14.0, 1.383578431372549)
(15.0, 1.383578431372549)
(16.0, 1.4018164735358596)
(17.0, 1.4176)
(18.0, 1.4362328319162851)
(19.0, 1.4362328319162851)
(20.0, 1.456093489148581)
(21.0, 1.4720708446866484)
(22.0, 1.4880249913224575)
(23.0, 1.4880249913224575)
(24.0, 1.5)
(25.0, 1.5125899280575539)
(26.0, 1.5267203513909224)
(27.0, 1.5379464285714286)
(28.0, 1.5379464285714286)
(29.0, 1.548960302457467)
(30.0, 1.5572196620583718)
        };

        \addplot[
            color=blue,
            mark=triangle*,
        ]
        coordinates {
            (1.0, 1.0760482817637256)
(2.0, 1.1254372813593203)
(3.0, 1.1254372813593203)
(4.0, 1.1661822788583351)
(5.0, 1.1990216271884655)
(6.0, 1.2288518833794424)
(7.0, 1.2285250784532737)
(8.0, 1.2579389019908955)
(9.0, 1.2830832581227436)
(10.0, 1.3055603822762816)
(11.0, 1.3055603822762816)
(12.0, 1.3262367123938477)
(13.0, 1.3458120282444608)
(14.0, 1.367550390977912)
(15.0, 1.367550390977912)
(16.0, 1.390923566878981)
(17.0, 1.4100045534378456)
(18.0, 1.4276656486522374)
(19.0, 1.4276656486522374)
(20.0, 1.4399701967690588)
(21.0, 1.4554021401449775)
(22.0, 1.4685885041307787)
(23.0, 1.4685885041307787)
(24.0, 1.4782748632121017)
(25.0, 1.4898367332097016)
(26.0, 1.50471436494731)
(27.0, 1.5173307296556386)
(28.0, 1.5173307296556386)
(29.0, 1.5287842420086104)
(30.0, 1.5392888923279298)
        };

        \addplot[
            color=green,
            mark=square*,
        ]
        coordinates {
            (1.0, 1.0716653384613608)
(2.0, 1.1202672605790647)
(3.0, 1.1203627069133397)
(4.0, 1.1618528679571203)
(5.0, 1.195073515873642)
(6.0, 1.2239057467828072)
(7.0, 1.2239057467828072)
(8.0, 1.246476904149424)
(9.0, 1.2737830030038662)
(10.0, 1.2982681300525167)
(11.0, 1.2982681300525167)
(12.0, 1.3192868474729924)
(13.0, 1.337461348175634)
(14.0, 1.3567453918871921)
(15.0, 1.3567453918871921)
(16.0, 1.3730446995011865)
(17.0, 1.3909075934771866)
(18.0, 1.4089197458225464)
(19.0, 1.4089197458225464)
(20.0, 1.4238737426529808)
(21.0, 1.4397076523225594)
(22.0, 1.4547262268304808)
(23.0, 1.4547262268304808)
(24.0, 1.4712234762979683)
(25.0, 1.4840539796199395)
(26.0, 1.4987405071557853)
(27.0, 1.5134402487109493)
(28.0, 1.5134402487109493)
(29.0, 1.5227902679396366)
(30.0, 1.5317506838608832)
        };

        \addplot[
            color=magenta,
            mark=square,
        ]
        coordinates {
            (1.0, 1.0692752320345706)
(2.0, 1.1171799887279117)
(3.0, 1.1171614699885624)
(4.0, 1.1551447500086263)
(5.0, 1.1850566818591293)
(6.0, 1.2133109578431012)
(7.0, 1.2149689762150981)
(8.0, 1.2373369730057628)
(9.0, 1.2645206574509154)
(10.0, 1.288618292244376)
(11.0, 1.289236458125125)
(12.0, 1.308989109073043)
(13.0, 1.329826616192249)
(14.0, 1.350539126721469)
(15.0, 1.3505247088213002)
(16.0, 1.371549222006416)
(17.0, 1.3915463481426233)
(18.0, 1.3996596130935497)
(19.0, 1.3996596130935497)
(20.0, 1.4133883285385833)
(21.0, 1.4303298925439112)
(22.0, 1.4443886517700226)
(23.0, 1.4443886517700226)
(24.0, 1.456487120576687)
(25.0, 1.4683090760829276)
(26.0, 1.4806507085744411)
(27.0, 1.4927624872579002)
(28.0, 1.4921160557352828)
(29.0, 1.5054859745367974)
(30.0, 1.513058207171838)
        };

        \addplot[
            color=orange,
            mark=heart,
        ]
        coordinates {
            (1.0, 1.0670085366957256)
(2.0, 1.1155302429712126)
(3.0, 1.1152897461748685)
(4.0, 1.1496641119013502)
(5.0, 1.1831339737349118)
(6.0, 1.2058434792338302)
(7.0, 1.2068309070548713)
(8.0, 1.2372413096637838)
(9.0, 1.2571128282559139)
(10.0, 1.2818080950861035)
(11.0, 1.2818080950861035)
(12.0, 1.3038458550229197)
(13.0, 1.3215699967463674)
(14.0, 1.332783558498468)
(15.0, 1.3339591189370277)
(16.0, 1.3508683593621025)
(17.0, 1.3709273816444658)
(18.0, 1.3864211185129933)
(19.0, 1.3857201617878934)
(20.0, 1.4034965096062277)
(21.0, 1.4175442646823775)
(22.0, 1.4327480626456404)
(23.0, 1.4332658703626737)
(24.0, 1.4487859104982224)
(25.0, 1.4624605560430943)
(26.0, 1.4758328354812096)
(27.0, 1.4864706108928434)
(28.0, 1.4864706108928434)
(29.0, 1.491677574384121)
(30.0, 1.5059185394651982)
        };

        \addplot[
            color=cyan,
            mark=|,
        ]
        coordinates {
            (1.0, 1.0675355333693994)
(2.0, 1.1089259632593447)
(3.0, 1.1106280840192433)
(4.0, 1.1429764409256613)
(5.0, 1.1735040559533427)
(6.0, 1.2021257850910847)
(7.0, 1.2032405616366657)
(8.0, 1.227253857532933)
(9.0, 1.2511544781758646)
(10.0, 1.274694694940659)
(11.0, 1.2747416790554238)
(12.0, 1.2920377259234208)
(13.0, 1.3108731687648403)
(14.0, 1.33304743510001)
(15.0, 1.3330998042284046)
(16.0, 1.348980137764997)
(17.0, 1.363751911639762)
(18.0, 1.3808098859842346)
(19.0, 1.381068179771233)
(20.0, 1.3987951637251168)
(21.0, 1.413131812611528)
(22.0, 1.43226093813891)
(23.0, 1.43226093813891)
(24.0, 1.4464146603850576)
(25.0, 1.4536136780951807)
(26.0, 1.4712533757679154)
(27.0, 1.482891091260448)
(28.0, 1.4834437086092715)
(29.0, 1.4933034485470489)
(30.0, 1.5091889732321215)
        };

            \legend{$Z = 1$, $Z = 10$, $Z = 20$, $Z = 30$, $Z = 40$, $Z = 50$}
        \end{axis}
    \end{tikzpicture}
    \caption{Empirical competitive ratio of LRU (as a function of $k$) over a WikiBench dataset. Results for other datasets were similar.}
    \label{fig:CR_v_k}
\end{figure}
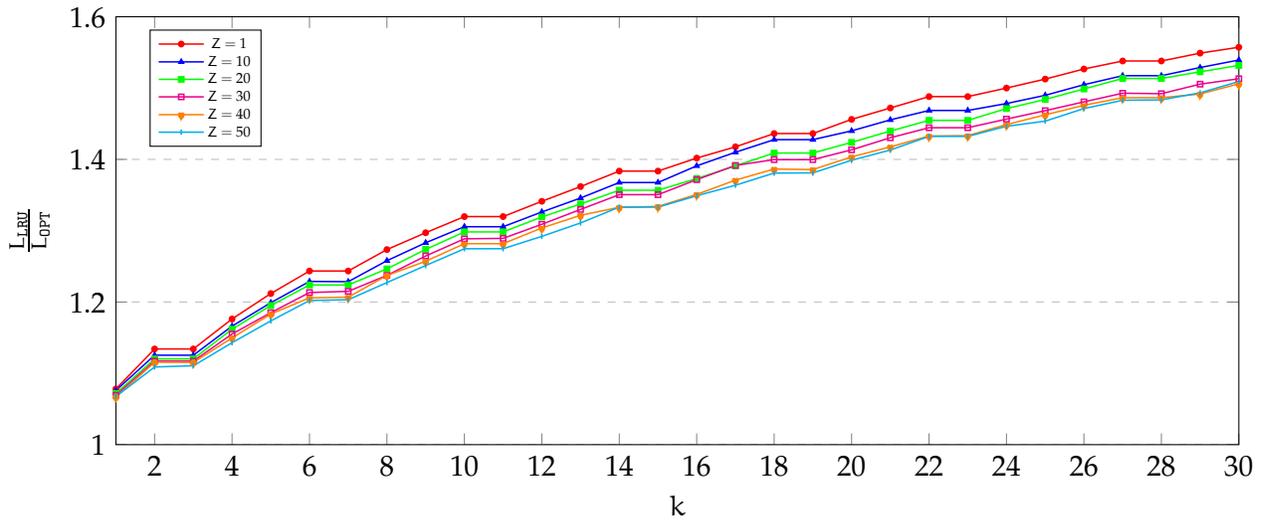
\begin{figure}
     \centering
    \input{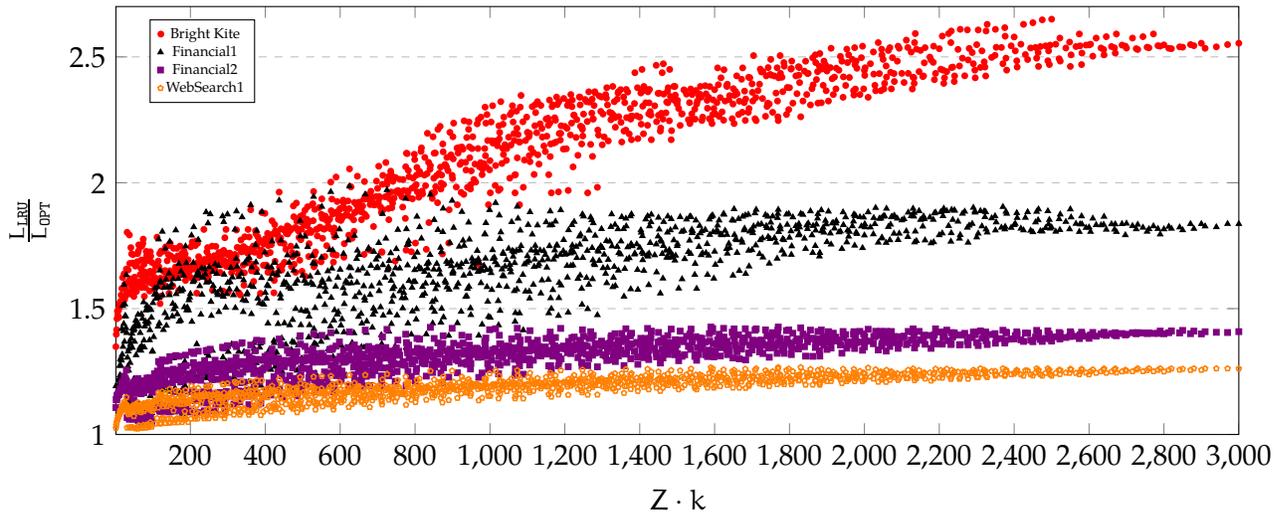}
    \caption{Empirical competitive ratio as a function of $Zk$. Only a few datasets are displayed to avoid cluttering, but the observed behavior was consistent across all datasets.}
    \label{fig:cloud_points}
\end{figure}

\paragraph{Experimental Setup} 
All experiments were run on a MacBook Pro M1 2020 with 16GB of RAM. The implementation of both LRU and ``OPT'' are taken from Atre et al.~\cite{ASWB20}, where ``OPT'' is in fact a \emph{near-optimal} solution computed offline using \texttt{Gurobi Optimizer}; to the best of our knowledge, it is not known whether optimal offline solutions can be computed in polynomial time. We verified that in each of our experiments, the lower bound and upper bound on the optimal solution obtained by \texttt{Gurobi Optimizer} always differed by less than $10\%$ from each other, and most of the time they were equal, thus certifying an optimal solution.

\paragraph{Interpretation of results}

As shown in~\Cref{fig:LRU_v_Z}, for a fixed $k$ the performance of LRU increases linearly with respect to $Z$, which implies that the average penalty when a page is not in the cache (either a ``miss'' or a ``delayed hit'') is $\Theta(Z)$. However, as $Z$ affects OPT (or any reference algorithm) as well, we observe in~\Cref{fig:CR_v_Z} that the empirical competitive ratio is actually slightly decreasing on $Z$. 
As we observe in~\Cref{fig:CR_v_k}, the competitive ratio does seem to increase with $k$ in practice, albeit at a rate that is much slower than linear.

Finally, the main conclusion of our experiments, visible in~\Cref{fig:cloud_points} is that the empirical competitive ratio, as a function of $Zk$, behaves much better than the pessimistic worse-case analysis of~\Cref{thm:lru-bound} would suggest. The next section presents a contextual discussion of these observation.

\section{Discussion}
\label{sec:discussion}

For the general delayed-hits caching model, our worst-guarantee competitive ratio guarantee of $O(Zk)$ is \emph{optimal} given the $\Omega(Zk)$ lower bound of Manohar and Williams~\cite{MW20}. However, our experimental results show that on practical data, the true competitive ratio is often far below this threshold, and in particular does not seem to depend on $Z$ at all. Is the dependency on $Z$ in~\Cref{thm:lru-bound} only an artifact of the analysis? (Roughly, the $Z$ factor in the analysis comes from upper-bounding the total loss from requests for a missed page in a phase by $Z^2$ --- the worst-case being a sequence of $Z$ consecutive requests for the same page --- while lower-bounding the loss by $Z$ --- the worst-case being only a single missed request.)

This poses tantalizing challenges: can one characterize assumptions on the request sequence that would yield a competitive ratio that \emph{does not} depend on $Z$? For instance, the \cite{MW20} lower bound of $\Omega(Zk)$ uses \emph{bursty} request sequences; is it possible to overcome this lower bound by making burstiness assumptions (e.g., no page gets requested more than $k$ time per phase/superphase)? Or, are there reasonable stochastic models where marking algorithms can be shown to have a competitive ratio?  In traditional caching, the competitive ratio of LRU is also observed to be constant in a variety of applications, and theoretical measures such as \emph{``locality''}, guarantee that certain hypothesis on the data distribution do indeed guarantee $O(1)$-competitiveness for LRU~\cite{PerformanceMeasures, DORRIGIV20093694, Becchetti}. Extending these results to the delayed hits model is a natural next step.

There are other interesting theoretical questions as well. Regarding our LRU analysis, could it be extended to heterogeneous page sizes, weights, or even delays? Or does it extend to more sophisticated classes of algorithms, like LRU-MAD considered in \cite{ASWB20}? Could our techniques prove a $O(Z \log k)$ competitive ratio bound for \emph{randomized} marking algorithms, and do the techniques of \cite{MW20} prove a $\Omega(Z \log k)$ lower bound? Recall that \cite{FKL+91} proved a $\Theta(\log k)$ bound for randomized algorithms in classical caching. We conjecture that $\Theta(Z \log k)$ is the correct threshold for the randomized case, just as our algorithm and the \cite{MW20} lower bound show that $\Theta(Zk)$ is the correct threshold for the deterministic case, where the classical threshold was $\Theta(k)$. Finally, what is the complexity of calculating the offline optimal value? Is it \textbf{NP}-hard in general, or is there a polynomial-time algorithm? What about in parameterized settings, such as $Z=2$?

\bibliographystyle{alpha}
\bibliography{delayedhits.bib}

\end{document}